\documentclass[11pt, article, one side]{eptcs}

\usepackage{comment}
 \usepackage[normalem]{ulem}
\usepackage{mathtools}
\usepackage{amsthm}
\usepackage{amssymb}
\usepackage[frozencache]{minted}
\usepackage[usenames,dvipsnames]{xcolor}
\usepackage{tikz}
\usepackage{outline}
\usepackage[inline]{enumitem}
\usepackage{ifthen}
\usepackage[utf8]{inputenc} 
\usepackage[backend=biber, backref=true, maxbibnames = 10]{biblatex}
\allowdisplaybreaks

  \makeindex 

  \hypersetup{final}

  \setlist{nosep}

  \setminted{fontsize=\footnotesize}

  \usetikzlibrary{ 
  	cd,
  	math,
  	decorations.markings,
  	decorations.pathmorphing,
		decorations.pathreplacing,
  	positioning,
  	arrows.meta,
  	shapes,
		shadings,
  	calc,
  	fit,
  	quotes,
  	intersections
  }
  
	\tikzcdset{arrow style=tikz, diagrams={>=To}}

\tikzset{
     oriented WD/.style={
        every to/.style={out=0,in=180,draw},
        label/.style={
           font=\everymath\expandafter{\the\everymath\scriptstyle},
           inner sep=0pt,
           node distance=2pt and -2pt},
        semithick,
        node distance=1 and 1,
        decoration={markings, mark=at position \stringdecpos with \stringdec},
        ar/.style={postaction={decorate}},
        execute at begin picture={\tikzset{
           x=\bbx, y=\bby,
           }}
        },
     string decoration/.store in=\stringdec,
     string decoration={\arrow{stealth};},
     string decoration pos/.store in=\stringdecpos,
     string decoration pos=.7,
     bbx/.store in=\bbx,
     bbx = 1.5cm,
     bby/.store in=\bby,
     bby = 1.5ex,
     bb port sep/.store in=\bbportsep,
     bb port sep=1.5,
     bb port length/.store in=\bbportlen,
     bb port length=4pt,
     bb penetrate/.store in=\bbpenetrate,
     bb penetrate=0,
     bb min width/.store in=\bbminwidth,
     bb min width=1cm,
     bb rounded corners/.store in=\bbcorners,
     bb rounded corners=2pt,
     bb small/.style={bb port sep=1, bb port length=2.5pt, bbx=.4cm, bb min width=.4cm, 
bby=.7ex},
		 bb medium/.style={bb port sep=1, bb port length=2.5pt, bbx=.4cm, bb min width=.4cm, 
bby=.9ex},
     bb/.code 2 args={
        \pgfmathsetlengthmacro{\bbheight}{\bbportsep * (max(#1,#2)+1) * \bby}
        \pgfkeysalso{draw,minimum height=\bbheight,minimum width=\bbminwidth,outer 
sep=0pt,
           rounded corners=\bbcorners,thick,
           prefix after command={\pgfextra{\let\fixname\tikzlastnode}},
           append after command={\pgfextra{\draw
              \ifnum #1=0{} \else foreach \i in {1,...,#1} {
                 ($(\fixname.north west)!{\i/(#1+1)}!(\fixname.south west)$) +(-
\bbportlen,0) 
  coordinate (\fixname_in\i) -- +(\bbpenetrate,0) coordinate (\fixname_in\i')}\fi 
              \ifnum #2=0{} \else foreach \i in {1,...,#2} {
                 ($(\fixname.north east)!{\i/(#2+1)}!(\fixname.south east)$) +(-
\bbpenetrate,0) 
  coordinate (\fixname_out\i') -- +(\bbportlen,0) coordinate (\fixname_out\i)}\fi;
           }}}
     },
     bb name/.style={append after command={\pgfextra{\node[anchor=north] at 
(\fixname.north) {#1};}}}
  }

  \newtheorem{theorem}{Theorem}[section]

  \newtheorem{lemma}[theorem]{Lemma}
  
  \theoremstyle{definition}
  \newtheorem{definition}[theorem]{Definition}

  \newtheorem*{axiom*}{Axiom}
  
  \theoremstyle{remark}
  \newtheorem{example}[theorem]{Example}

\DeclareSymbolFont{stmry}{U}{stmry}{m}{n}
\DeclareMathSymbol\fatsemi\mathop{stmry}{"23}

\DeclareFontFamily{U}{mathx}{\hyphenchar\font45}
\DeclareFontShape{U}{mathx}{m}{n}{
      <5> <6> <7> <8> <9> <10>
      <10.95> <12> <14.4> <17.28> <20.74> <24.88>
      mathx10
      }{}
\DeclareSymbolFont{mathx}{U}{mathx}{m}{n}
\DeclareFontSubstitution{U}{mathx}{m}{n}
\DeclareMathAccent{\widecheck}{0}{mathx}{"71}

\DeclareMathOperator{\Hom}{Hom}

\DeclareMathOperator{\ob}{Ob}

\newcommand{\cat}[1]{\mathcal{#1}}
\newcommand{\Cat}[1]{\mathsf{#1}}

\newcommand{\id}{\mathrm{id}}

\newcommand{\from}{\leftarrow}

\newcommand{\LMO}[2][over]{\ifthenelse{\equal{#1}{over}}{\overset{#2}{\bullet}}{\underset{#2}{\bullet}}}

\newcommand{\fork}{\triangledown}

\newcommand{\inp}{_-}
\newcommand{\outp}{_+}

\newcommand{\src}{_{\text{src}}}
\newcommand{\tgt}{_{\text{tgt}}}

\newcommand{\smset}{\Cat{Set}}

\newcommand{\erase}[1]{}

\newcommand{\bij}{\Cat{Bij}}

\newcommand{\qqand}{\qquad\text{and}\qquad}

\newcommand{\dsnote}[1]{\textnormal{\color{blue}David says: #1}}
\newcommand{\dvnote}[1]{\textnormal{\color{green!50!black}Dmitry says: #1}}

\newcommand{\one}{\mathbf{1}}
\newcommand{\zero}{\textbf{0}}

\newcommand{\bfs}{\mathbf{s}}
\newcommand{\bft}{\mathbf{t}}
\newcommand{\bfu}{\mathbf{u}}
\newcommand{\bfv}{\mathbf{v}}

\newcommand{\then}{\mathbin{\fatsemi}}

\newcommand{\seq}{\mathsf{seq}}
\newcommand{\para}{\mathsf{para}}
\newcommand{\unit}{\mathsf{unit}}
\newcommand{\inert}{\mathsf{inert}}
\newcommand{\sym}{\mathsf{sym}}

\title{Wiring diagrams as normal forms for\\ computing in symmetric monoidal categories}
\author{Evan Patterson \and David I.\ Spivak \and Dmitry Vagner}

\def\authorrunning{Evan Patterson, David I. Spivak, and Dmitry Vagner}

\def\titlerunning{Wiring diagrams as normal forms for computing in SMCs}

\begin{document}

\maketitle


\begin{abstract}
\noindent Applications of category theory often involve symmetric monoidal categories (SMCs), in which abstract processes or operations can be composed in series and parallel. However, in 2020 there remains a dearth of computational tools for working with SMCs. We present an  ``unbiased'' approach to implementing symmetric monoidal categories, based on an operad of directed, acyclic wiring diagrams. Because the interchange law and other laws of a SMC hold identically in a wiring diagram, no rewrite rules are needed to compare diagrams. We discuss the mathematics of the operad of wiring diagrams, as well as its implementation in the software package Catlab.
\end{abstract}

\section{Introduction}

The syntax for an algebraic structure is often derived from its traditional axiomatization, without additional thought. A symmetric monoidal category (SMC) is defined through operations of composition, identity, monoidal product, monoidal unit, and braiding, subject to various laws. Once it is decided how to assign symbols to these operations, such as $\then$ for composition and $\otimes$ for the monoidal product, a symbolic syntax for constructing objects and morphisms follows immediately. So, given morphisms, say $f\colon x \to x \otimes y$ and $g\colon y \otimes z \to z$, a new morphism can be constructed via such expressions as
\begin{equation*} 
(f \otimes \id_z) \then (\id_x \otimes g).
\end{equation*}
Symbolic syntax has a long tradition in algebra. Its utility derives, on the one hand, from its ease in writing and typesetting, and on the other, from its immediacy given an axiomatization of an algebraic structure as a generalized algebraic theory.

But these are not the only desiderata for mathematical syntax. In general, the mathematical objects denoted by two different expressions may be equal under the axioms. A good syntax narrows the gap between a mathematical object and its representation by avoiding redundancy. For example, since  monoidal products are associative in a strict SMC, the expressions $f \otimes (g \otimes h)$ and $(f \otimes g) \otimes h$ denote the same morphism; thus, it is standard practice to eliminate parentheses around the monoidal product, writing simply $f \otimes g \otimes h$.

Encoding algebraic equations into a simplified yet unambiguous syntax has important cognitive and computational benefits. For humans, it substitutes visual inspection for equational reasoning, playing to our cognitive strengths. For computers, it reduces possibly complex algorithms for checking \emph{equality} to a simple test of \emph{identity} on a suitable data structure. From this perspective, the ideal syntax provides a \emph{normal form}, making two expressions identical if and only if they denote equal mathematical objects.

What is the right syntax for symmetric monoidal categories? Beginning with the Penrose graphical notation for tensors \cite{penrose1971}, it was gradually understood that morphisms in a monoidal category are best depicted by a \emph{two-dimensional} syntax, with one axis representing composition and the other representing monoidal product. For example, the above expression $(f \otimes \id_z) \then (\id_x \otimes g)$ becomes the \emph{string diagram}
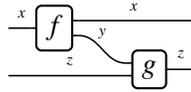
\[
\begin{tikzpicture}[oriented WD, bb small]
 	\node[bb={1}{2}] (Monf) {$f$};
  	\node[bb={2}{1}, below right=0 and 2 of Monf] (Mong) {$g$};
  	\node[bb={0}{0}, inner xsep=.3cm, white, fit=(Monf) (Mong)] (Mon) {};
  	\node[coordinate] at (Mon.west|-Monf_in1) (Mon_in1) {};
  	\node[coordinate] at (Mon.west|-Mong_in2) (Mon_in2) {};
  	\node[coordinate] at (Mon.east|-Monf_out1) (Mon_out1) {};
  	\node[coordinate] at (Mon.east|-Mong_out1) (Mon_out2) {};
  	\draw[shorten <=-2pt] (Mon_in1) -- node[above, font=\tiny] {$x$} (Monf_in1);
  	\draw[shorten >=-2pt] (Monf_out1) -- node[above, font=\tiny] {$x$} (Mon_out1);
  	\draw (Monf_out2) to node[above, font=\tiny] {$y$} (Mong_in1);
  	\draw[shorten <=-2pt] (Mon_in2) -- node[above, font=\tiny] {$z$} (Mong_in2);
  	\draw[shorten >=-2pt] (Mong_out1) -- node[above, font=\tiny] {$z$} (Mon_out2);
 \end{tikzpicture}
\]
As string diagrams, both sides of the \emph{interchange law}
$(f \then g) \otimes (h \then k) = (f \otimes h) \then (g \otimes k)
$
of a monoidal category have the same representation, namely:
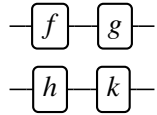
\[
\begin{tikzpicture}[oriented WD, bb small]
    \node[bb={1}{1}] (Monf) {$f$};
    \node[bb={1}{1}, right=of Monf] (Mong) {$g\vphantom{f}$};
    \node[bb={1}{1}, below=2 of Monf] (Monh) {$h\vphantom{f}$};
    \node[bb={1}{1}, right=of Monh] (Monk) {$k\vphantom{f}$};
    \node[bb={0}{0}, inner xsep=.3cm, white, fit=(Monf) (Mong) (Monh) (Monk)] (Mon) {};
    \node[coordinate] at (Mon.west|-Monf_in1) (Mon_in1) {};
  	\node[coordinate] at (Mon.west|-Monh_in1) (Mon_in2) {};
  	\node[coordinate] at (Mon.east|-Mong_out1) (Mon_out1) {};
  	\node[coordinate] at (Mon.east|-Monk_out1) (Mon_out2) {};
  	\draw (Mon_in1) -- (Monf_in1);
  	\draw (Mon_in2) -- (Monh_in1);
  	\draw (Monf_out1) -- (Mong_in1);
  	\draw (Monh_out1) -- (Monk_in1);
  	\draw (Mong_out1) -- (Mon_out1);
  	\draw (Monk_out1) -- (Mon_out2);
\end{tikzpicture}
\]
String diagrams were first put on a rigorous footing by Joyal and Street, who showed that the diagrammatic language is sound and complete for the equations between morphisms deducible from the axioms of a strict symmetric monoidal category \cite{joyal1991}. Diagrammatic languages are now known for other kinds of monoidal categories, such as traced monoidal categories and hypergraph categories \cite{selinger2010,fong2019}.

For Joyal and Street, string diagrams are geometric figures in the plane or in a higher-dimensional Euclidean space. This perspective, while intuitively appealing, is of little computational use, since geometric objects are not readily translated into data structures. For computational purposes, we would like to extract the combinatorial data defining a string diagram, much like a graph does for a graph embedding or graph drawing.

\emph{Wiring diagrams} were introduced in \cite{rupel2013}. Though they are often depicted graphically, wiring diagrams are combinatorial, rather than geometric, objects. They also differ from string diagrams in that they deal only with the syntax of composition, and do not include explicit morphisms from a given SMC.  More precisely, wiring diagrams---say for representing compositions in symmetric monoidal categories---are organized as the morphisms of a typed operad $\cat{W}$. Composition of morphisms corresponds to nesting of wiring diagrams, for example:
\begin{equation}\label{eqn.nesting}
\raisebox{-.25in}{
\begin{tikzpicture}
\node (p1) {
\begin{tikzpicture}[oriented WD, bb small]
    \node[bb={1}{1}, orange] (a2) {};
    \node[bb={1}{1}, orange, below=of a2] (a3) {};
    \coordinate (helper) at ($(a2)!.5!(a3)$);
    \node[bb port sep=.8, bb={3}{2}, red, left=1.5 of helper] (a1) {};
    \node[bb={0}{0}, blue, fit=(a1) (a2) (a3)] (outer a) {};
    \draw (a1_in1) -- (outer a.west|-a1_in1);
    \draw (a1_in2) -- (outer a.west|-a1_in2);
    \draw (a1_in3) -- (outer a.west|-a1_in3);
    \draw (a1_out1) to (a2_in1);
    \draw (a1_out2) to (a3_in1);
    \draw (a2_out1) to (a2_out1-|outer a.east);
    \draw (a3_out1) to (a3_out1-|outer a.east);
\end{tikzpicture}
};
\node (p2) [right=.5 of p1] {
\begin{tikzpicture}[oriented WD, bb small]
    \node[bb={3}{2}, blue] (b1) {\tiny 1};
    \node[bb={2}{1}, green!50!black, below right=0 and 1 of b1] (b2) {\tiny 2};
    \node[bb={0}{0}, fit=(b1) (b2)] (outer b) {};
    \draw (b1_in1) -- (b1_in1-|outer b.west);
    \draw (b1_in2) -- (b1_in2-|outer b.west);
    \draw (b1_in3) -- (b1_in3-|outer b.west);
    \draw (b2_in2) -- (b2_in2-|outer b.west);
    \draw (b1_out1) -- (b1_out1-|outer b.east);
    \draw (b2_out1) -- (b2_out1-|outer b.east);
    \draw (b1_out2) to (b2_in1);
\end{tikzpicture}
};
\node (p3) [right=of p2] {
\begin{tikzpicture}[oriented WD, bb small]
    \node[bb={1}{1},orange] (c2) {};
    \node[bb={1}{1},orange, below=of c2] (c3) {};
    \coordinate (helper) at ($(c2)!.5!(c3)$);
    \node[bb port sep=.8, bb={3}{2}, red, left=1.5 of helper] (c1) {};
    \node[bb={2}{1}, green!50!black, below right=0 and 1 of c3] (c4) {\tiny 2};
    \node[bb={0}{0}, fit=(c1) (c2) (c3) (c4)] (outer c) {};
    \draw (c1_in1) -- (outer c.west|-c1_in1);
    \draw (c1_in2) -- (outer c.west|-c1_in2);
    \draw (c1_in3) -- (outer c.west|-c1_in3);
    \draw (c1_out1) to (c2_in1);
    \draw (c1_out2) to (c3_in1);
    \draw (c2_out1) to (c2_out1-|outer c.east);
    \draw (c3_out1) to (c4_in1);
    \draw (c4_in2) -- (c4_in2-|outer c.west);
    \draw (c4_out1) -- (c4_out1-|outer c.east);
\end{tikzpicture}
};
\node at ($(p1.east)!.5!(p2.west)$) {$\then_{\color{blue}1}$};
\node at ($(p2.east)!.5!(p3.west)$) {$=$};
\end{tikzpicture}
}
\end{equation}
Note that in the symbolic syntax of a symmetric monoidal category, there are infinitely-many ways to represent any of the wiring diagrams in \eqref{eqn.nesting}, for example the middle one could be represented by tensoring with an arbitrary number of monoidal units, but the wiring diagram represents exactly one element of the set $\cat{W}({\color{blue}X_1},{\color{green!50!black}X_2};Y)$, where ${\color{blue}X_1}$ and ${\color{green!50!black}X_2}$ are the inner boxes and $Y$ is the outer box.

A given model of $\cat{W}$, which we can think of as a symmetric monoidal category $\cat{C}$, is represented as a functor $H\colon\cat{W}\to\smset$, which we refer to as a $\cat{W}$-algebra.%
\footnote{For now, we elide the fact that wires should be annotated by types corresponding to the objects of $\cat{C}$.} 
Each box (object in $\cat{W}$) is sent by $H$ to the set of all $\cat{C}$-morphisms of that shape, and each wiring diagram (morphism in $\cat{W}$) is sent by $H$ to the function that takes morphisms in $\cat{C}$ and composes them accordingly.

Operads of wiring diagrams can be used to model various sorts of categories, such as traced or hypergraph categories \cite{spivak2017,fong2019}---whose string diagram languages we mentioned above---or even just ordinary categories. Here we will focus on the operad $\cat{W}$ of wiring diagrams for symmetric monoidal categories.


The aim of this paper is to show how $\cat{W}$ can serve as a foundation for the computer algebra of symmetric monoidal categories, both practically and theoretically. Since the interchange law and other axioms of an SMC hold identically within $\cat{W}$, wiring diagrams are nearly normal forms for morphisms in a free SMC, as will be explained. With this motivation, Section~\ref{wiring-diagrams-in-catlab} describes an implementation of the operad of wiring diagrams in Catlab.jl, a Julia package for applied category theory \cite{catlab}. The operad structure is taken as fundamental in Catlab and used to implement a diagrammatic syntax for symmetric monoidal categories. The wiring diagram operad is then extended to SMCs with extra structure, such as cartesian monoidal categories, biproduct categories, and traced monoidal categories. In contrast to systems like Quantomatic \cite{quantomatic} and Cartographer \cite{cartographer}, Catlab is not a graphical editor or a proof assistant; rather, it provides data structures and algorithms for computing with in symmetrical monoidal categories for scientific and engineering applications.

The remainder of the paper is about the mathematics of the operad of wiring diagrams. After some preliminaries on biproducts and spans in Section~\ref{ch:spans}, an operad of wiring diagrams is constructed in Section~\ref{ch:operads}. The directed wiring diagrams are defined to satisfy an acyclicity condition, making them suitable for symmetric monoidal categories. They are built using the categorical matrix calculus, which extends the notion of an adjacency matrix of a directed graph. Finally, Section~\ref{ch:smcs} shows how algebras on the wiring diagram operad give rise to symmetric monoidal categories and, conversely, how symmetric monoidal categories yield algebras on this operad. These two constructions are not inverse equivalences; however, beginning with an SMC, the roundtrip does return an equivalent SMC. We will make the correspondence more precise in Theorems~\ref{thm.wdalgto_smc} and \ref{thm.smc_alg}.
\section{Wiring diagrams in Catlab}
\label{wiring-diagrams-in-catlab}

In the Julia package Catlab, wiring diagrams are implemented as combinatorial data structures. Wiring diagrams are akin to directed graphs but possess extra structure; namely each box---which plays the role of a node---has an explicit set of input and output ports. Every wiring diagram has an underlying directed graph, obtained by forgetting this extra structure. For implementation purposes, Catlab exploits this hierarchy by building its data structure for wiring diagrams on top of graph data structures that already exist in the Julia ecosystem.

The functionality for wiring diagrams is implemented in several layers. The bottom layer is the core data structure and a low-level imperative interface for mutating it. Operadic composition---or substitution---of wiring diagrams is defined using this interface. Finally, the operadic interface is used to define a syntax for morphisms in symmetric monoidal categories. The three layers are summarized in Sections~\ref{sec.wd_data_struc}, \ref{sec:catlab compose}, and \ref{sec.wd_syntax} respectively.

\subsection{The wiring diagram data structure}\label{sec.wd_data_struc}

The data structure for wiring diagrams is comprised of several data types. From the user's perspective, these types are:
\begin{enumerate}
\item \texttt{WiringDiagram}: a quadruple consisting of (i) a list of input port types and (ii) a list of output port types, both for the outer box; (iii) a list of boxes, where the first and second entries are special values representing the input and output types of the outer box and the remaining entries are boxes inside the diagram, of type \texttt{Box}; and (iv) a set of wires, of type \texttt{Wire}.

\item \texttt{Box}: a triple consisting of (i) a label or value for the box, (ii) a list of input port types, (iii) and a list of output port types. Thus, a morphism \(f: x \otimes y \to z\) would be represented as \texttt{Box(:f, [:x,:y], [:z])}, where the Julia syntax \texttt{:x} denotes a symbol named ``\texttt{x}''.

\item \texttt{Wire}: a source-target pair, where both the source and target are pairs of numbers identifying a box in the diagram and an input or output port on that box. So, a wire from output port $p$ on box $i$ to input port $q$ on box $j$ is \texttt{Wire((i,p) => (j,q))}.
\end{enumerate}

As a small but complete example, the wiring diagram corresponding to the composite $f \then g: x \to z$ of morphisms $f: x \to y$ and $g: y \to z$ is implemented as:
\begin{center}
\begin{minted}{julia}
WiringDiagram([:x], [:z],
  [ 1 => {inputs}, 2 => {outputs}, 3 => Box(:f, [:x], [:y]), 4 => Box(:g, [:y], [:z]) ],
  [ Wire((1,1) => (3,1)), Wire((3,1) => (4,1)), Wire((4,1) => (2,1)) ])
\end{minted}
\end{center}

For performance reasons, the wires in a wiring diagram are not actually stored as a set. Instead, underlying each wiring diagram is a simple directed graph, as implemented by the Julia package LightGraphs.jl \cite{bromberger17}. The vertices in the graph are numbered consecutively from 1 to $n+2$, where $n$ is the number of boxes in the wiring diagram. Vertices 1 and 2, labelled \texttt{\{inputs\}} and \texttt{\{outputs\}} above, refer to the inputs and outputs of the outer box, respectively; the remaining vertices refer to boxes inside the diagram. There is an edge between two vertices if and only if there is at least one wire between the corresponding boxes.

Traversals on the wiring diagram are then delegated to the underlying directed graph. For example, to find the in-neighbors or out-neighbors of a box, one simply finds the in-neighbors or out-neighbors of the corresponding vertex in the underlying graph. The graph data structure is designed to do this efficiently. Wiring diagrams can be accessed and mutated through a simple imperative interface. Boxes, ports, and wires can be retrieved individually or iterated over, and boxes and wires can be added and removed from an existing diagram.

\subsection{Implementing the operad of wiring diagrams}
\label{sec:catlab compose}

Operadic composition of wiring diagrams as in \eqref{eqn.nesting} is supported through the function \texttt{ocompose}, with two signatures:

\begin{minted}{Julia}
ocompose(f::WiringDiagram, gs::Vector{<:WiringDiagram})::WiringDiagram
ocompose(f::WiringDiagram, i::Int, g::WiringDiagram)::WiringDiagram
\end{minted}

The first signature corresponds to full (May-style) operadic composition \(\circ\) and the second corresponds to partial (Markl-style) operadic composition \(\circ_i\). Both methods are one-line wrappers around the procedure \texttt{substitute}, which substitutes wiring diagrams for one or more boxes in another wiring diagram. Mathematically, \texttt{substitute} simultaneously performs one or more non-overlapping partial operadic compositions.

Substitution of wiring diagrams is implemented using the imperative interface. In outline, the algorithm proceeds as follows:
\begin{enumerate}

\item Create a copy \texttt{d} of the original wiring diagram.

\item Add to \texttt{d} all the boxes from all the diagrams to be substituted.

\item Extend each wire in each substituted diagram to new a wire in \texttt{d}. This subroutine branches into four different cases, shown in pseudo-Julia code in Listing \ref{lst:substitute-wires}. 

\item Remove from \texttt{d} all the boxes from the original diagram that were to be substituted, which in turn removes any extraneous wires created during step 3.

\end{enumerate}

\begin{listing}[h]
\begin{minted}{julia}
function substitute_wires!(d::WiringDiagram, v::Int, sub::WiringDiagram)
  for wire in wires(sub)
    # Case 1: Passing wire.
    if {wire source and target are on outer box}
      for in_wire in {wires in d incoming to source outer port}
        for out_wire in {wires in d outgoing from target outer port}
          {add wire to d fusing in_wire -> wire -> out_wire}
        end
      end
    # Case 2: Incoming wire.
    elseif {wire source is on outer box}
      for in_wire in {wires in d incoming to source outer port}
        {add wire to d fusing in_wire -> wire}
      end
    # Case 3: Outgoing wire.
    elseif {wire target is on outer box}
      for out_wire in {wires in d outgoing from target outer port}
        {add wire to d fusing wire -> out_wire}
      end
    # Case 4: Fully internal wire.
    else
      {add wire to d}
    end
  end
end
\end{minted}
\caption{Pseudo-Julia code for main subroutine in substitution algorithm. Each wire substituted into the new diagram is either passing, incoming, outgoing, or fully internal. The inner loops under each of the cases are needed because a port may have many or no incident wires, representing copying, merging, deleting, or creating.}
\label{lst:substitute-wires}
\end{listing}

\subsection{Wiring diagrams as a syntax for SMCs}\label{sec.wd_syntax}

Having implemented the operad of wiring diagrams, it is now straightforward to define an alternate syntax for symmetric monoidal categories using wiring diagrams. That is, we construct a symmetric monoidal category whose morphisms are wiring diagrams in which each box has been filled with a morphism. To compute the series composition $f\then g$ of morphisms $f$ and $g$, one forms the following wiring diagram
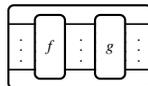
\[
\begin{tikzpicture}[oriented WD, bb small, bb port length=0]
    \node[bb={6}{6}] (f) {\tiny$f$};
    \node[bb={6}{6}, right=of f] (g) {\tiny$g$\vphantom{$f$}};
    \node[bb={0}{0}, inner xsep=10pt, fit=(f) (g)] (outer) {};
    \draw (f_out1) -- (g_in1);
    \draw (f_out6) -- (g_in6);
        \draw (f_in1) -- (f_in1-|outer.west);
        \draw (f_in6) -- (f_in6-|outer.west);
        \draw (g_out1) -- (g_in1-|outer.east);
        \draw (g_out6) -- (g_in6-|outer.east);
    \node[above left=-1pt and .5pt of f_in6] {\tiny $\vdots$};
    \node[above left=-1pt and .5pt of g_in6] {\tiny $\vdots$};
    \node[above right=-1pt and .5pt of g_out6] {\tiny $\vdots$};
\end{tikzpicture}
\]
and then performs an operadic composition. Apart from exception handling and formatting, the following  Julia code for the procedure \texttt{compose} is identical to the implementation in Catlab.

\begin{minted}{julia}
function compose(f::WiringDiagram, g::WiringDiagram)::WiringDiagram
  @assert length(codom(f)) == length(dom(g))
  h = WiringDiagram(dom(f), codom(g))
  fv, gv = add_box!(h, f), add_box!(h, g)
  add_wires!(h, [[ (input_id(h),i) => (fv,i) for i in 1:length(dom(f)) ];
                 [ (fv,i) => (gv,i) for i in 1:length(codom(f)) ];
                 [ (gv,i) => (output_id(h),i) for i in 1:length(codom(g)) ]])
  substitute(h, [fv,gv])
end
\end{minted}
\noindent
Similarly, to compute the parallel composition $f \otimes g$ of two morphisms $f$ and $g$, form the generic diagram with two boxes composed in parallel and then perform an operadic composition:
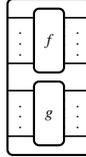
\[
\begin{tikzpicture}[oriented WD, bb small, bb port length=0, font=\tiny]
    \node[bb={6}{6}] (f) {$f$};
    \node[bb={6}{6}, below=of f] (g) {$g$\vphantom{$f$}};
    \node[bb={0}{0}, inner xsep=10pt, fit=(f) (g)] (outer) {};
        \draw (f_in1) -- (f_in1-|outer.west);
        \draw (f_in6) -- (f_in6-|outer.west);
        \draw (f_out1) -- (f_out1-|outer.east);
        \draw (f_out6) -- (f_out6-|outer.east);
        \draw (g_in1) -- (g_in1-|outer.west);
        \draw (g_in6) -- (g_in6-|outer.west);
        \draw (g_out1) -- (g_out1-|outer.east);
        \draw (g_out6) -- (g_out6-|outer.east);
    \node[above left=-1pt and .5pt of f_in6] {\tiny $\vdots$};
    \node[above right=-1pt and .5pt of f_out6] {\tiny $\vdots$};
    \node[above left=-1pt and .5pt of g_in6] {\tiny $\vdots$};
    \node[above right=-1pt and .5pt of g_out6] {\tiny $\vdots$};
\end{tikzpicture}
\]

The wiring diagram syntax provides a normal form for morphisms in a free symmetric monoidal category, up to (simple directed) graph isomorphism. Specifically, two morphisms represented by wiring diagrams are equal if and only if there is an isomorphism of the underlying graphs making the diagrams identical as Julia data structures. Although the graph isomorphism problem is not known to be solvable in polynomial time, in practice it can usually be solved efficiently. Moreover, the labels on the boxes drastically restrict the possible matchings. Wiring diagrams thus constitute an effective normal form for morphisms in a free SMC.
\section{Biproducts, matrix calculus, and categories of spans}\label{ch:spans}

We will formally present a wiring diagram by encoding all the interconnections between its boxes as a single span. Since the category of spans enjoys a biproduct structure, we can represent wiring diagrams as matrices and their compositions using matrix algebra. We begin with a brief detour into biproduct categories.

\subsection{Morphisms as matrices}
Let $(\cat{C},\oplus,\zero)$ be a biproduct category, i.e., a monoidal category for which $\oplus$ is both a product and a coproduct. Such categories can equivalently be seen as monoidal categories with a homomorphic supply of bimonoids \cite{fong2019supply}. 
Given a morphism %
\[\bft_1\oplus\cdots\oplus\bft_m\xrightarrow{f}\bfv_1\oplus\cdots\oplus\bfv_n\]
in such a category, we can extract the \emph{component} ${}_{\bft_i}f{}_{\bfv_j}$ corresponding to a choice of direct summand $\bft_i$ in the domain and $\bfv_j$ in the codomain, by pre- and post-composing $f$ with canonical inclusions and projections:
\begin{equation}\label{eqn.extract}
\begin{tikzcd}
    \bft_i \ar[r,"\iota_{\bft_i}"] 
&   \bft_1\oplus\cdots\oplus\bft_n \ar[r, "f"]
&   \bfv_1\oplus\cdots\oplus\bfv_m \ar[r,"\pi_{\bfv_j}"]
&   \bfv_j.
\end{tikzcd}
\end{equation}
This procedure defines a map of type $\cat{C}\big(\bigoplus_{i:I}\bft_i,\,\bigoplus_{j:J}\bfv_j\big)\to \cat{C}(\bft_i,\bfv_j)$ for finite sets $I$ and $J$,
to which we can define a section by swapping the inclusion and projections in \eqref{eqn.extract}. We say that this section is the \emph{embedding} of the component since it produces a morphism for which all other components are the zero morphisms $\zero:\bft_{k}\to\bfv_{l}$, i.e. the unique morphism of the form $\bft_k\to\zero\to\bfv_l$. 

Any biproduct category is enriched in commutative monoids, with the sum $f+g$ of two morphisms of type $\bft\to\bfv$ given by $\begin{tikzcd}
\bft \ar[r,"\Delta"] & \bft\oplus \bft \ar[r,"f\oplus g"] & \bfv \oplus \bfv \ar[r, "\nabla"] & \bfv
\end{tikzcd}
$.
Here $\Delta$ and $\nabla$ are the diagonal and codiagonal morphisms arising from the universal property of product and coproduct. By the naturality of these maps, composition distributes over sums. This operation allows us to define the leftward map in the bijection $
\cat{C}\big(\bigoplus_{i: I}\bft_i,\,\bigoplus_{j: J}\bfv_j\big) \cong \prod_{(i,j):I\times J}\cat{C}(\bft_i,\bfv_j)
$
by simply summing across the embeddings of each component. By the universal property of the biproduct, a morphism is fully specified by its components. More formally, given a choice of direct sum decomposition on both domain and codomain, we can represent $f$ as the block matrix 
\[
f = 
\begin{bmatrix}
{}_{\bft_1}f{}_{\bfv_1} & \cdots & {}_{\bft_1}f{}_{\bfv_n} \\
\vdots      & \ddots & \vdots      \\
{}_{\bft_m}f{}_{\bfv_1} & \cdots & {}_{\bft_m}f{}_{\bfv_n}
\end{bmatrix}
\]
meaning that we can reason about morphisms by simultaneously breaking up both domain and codomain into cases. As an example of this reasoning, we recover matrix multiplication from the composition of two morphisms
\[
\begin{tikzcd}
\bfs_1\oplus\cdots\oplus\bfs_k \ar[r, "f"] & \bft_1\oplus\cdots\oplus\bft_m \ar[r,"g"] & \bfv_1\oplus\cdots\oplus\bfv_n.
\end{tikzcd}
\]
The ${}_{\bfs_i}(f g){}_{\bfv_j}$ component is defined as the composite $(\iota_{\bfs_i}f)(g \pi_{\bfv_j})$, where $\iota_{\bfs_i}f$ is given by a sum of morphisms of type $\bfs_i\to\bft_1\oplus\cdots\oplus\bft_m$ and $g \pi_{\bfv_j}$ is given by a sum of morphisms of type $\bft_1\oplus\cdots\oplus\bft_m\to\bfv_j$:
\[    \iota_{\bfs_i}  f = ({}_{\bfs_i}f{}_{\bft_1})( \iota_{\bft_1})+\cdots+({}_{\bfs_i}f{}_{\bft_m}) (\iota_{\bft_m})
    \qqand
    g \pi_{\bfv_j} = (\pi_{\bft_1})( {}_{\bft_1}g{}_{\bfv_j})+\cdots+(\pi_{\bft_m})( {}_{\bft_m}g{}_{\bfv_j}).
\]
Composing these sums, applying distributivity, and then noting that $\iota_{\bft_k} \pi_{\bft_l}=\delta_{k,l}$ (the Kronecker delta), we arrive at the familiar formula for the component of a matrix multiplication
\[{}_{\bfs_i}(f  g){}_{\bfv_j}=\sum_{i,j}({}_{\bfs_i}f{}_{\bft_k})(\iota_{\bft_k} \pi_{\bft_l})({}_{\bft_l}g{}_{\bfv_j})=\sum_{k}({}_{\bfs_i}f{}_{\bft_k})( {}_{\bft_k}g{}_{\bfv_j}).\]
We can narrativize this expression as answering the question ``how can we go from $\bfs_i$ to $\bfv_j$?''
with the response ``through any one of the $\bft_k$,'' demonstrating that the sum can be interpreted as a disjunction. The capacity that biproduct categories possess for declaratively formalizing such reasoning motivates us to situate wiring diagrams in a certain biproduct category, which we now proceed to define.

\subsection{Spans}
Let $\cat{C}$ be an extensive category, i.e., one for which the coproduct functor
$\cat{C}/x\times\cat{C}/y\to \cat{C}/(x+y)$
is an equivalence of categories. For example, $\cat{C}$ could be any elementary topos. Recall the bicategory of $\cat{C}$-spans, whose objects are $\cat{C}$-objects, morphisms $X\to Y$ are spans $Y\leftarrow S\to X$ in $\cat{C}$, and 2-morphisms from $Y\leftarrow S\to X$ to $Y\leftarrow S'\to X$ are $\cat{C}$-arrows $S\to S'$ commuting with the legs of the spans. Let $\Cat{Sp}\cat{C}$ be the category whose objects are $\cat{C}$-objects and whose morphisms are isomorphism classes of spans in the bicategory of $\cat{C}$-spans. The coproduct in $\cat{C}$ is then the biproduct in $\Cat{Sp}\cat{C}$, which we denote as $\oplus$. This means that we can compute with spans using the matrix calculus of the previous subsection.

It is worthwhile to observe how the matrix calculus manifests itself in this concrete setting. Given a span $\Phi:\bft_1\oplus\cdots\oplus\bft_m\to\bfv_1\oplus\cdots\oplus\bfv_n$ and indices $1\leq i\leq m$ and $1\leq j\leq n$, we can define a component sub-span ${}_{\bft_i}\Phi{}_{\bfv_j}:\bft_i\to \bfv_j$ via the following limit:
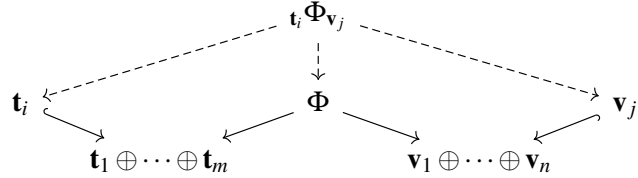
\[
\begin{tikzcd}[column sep=small, row sep=3pt]
{} & {} & {}_{\bft_i}\Phi{}_{\bfv_j} \ar[ddrr,dashed] \ar[dd,dashed] \ar[ddll,dashed] \arrow[d, phantom, very near start] & {} & {} \\
{} & {} & {} \\
\bft_i \ar[dr,hookrightarrow] & {} & \Phi \ar[dr] \ar[dl] & {} & \bfv_j \ar[dl,hookrightarrow] \\
{} & \bft_1\oplus\cdots\oplus\bft_m & {} & \bfv_1\oplus\cdots\oplus\bfv_n& {} 
\end{tikzcd}    
\]
The multiplication $\bft\oplus\bft\to\bft$ and  unit $\zero\to\bft$ morphisms are given by the spans
\[
\begin{tikzcd}
 \bft\oplus  \bft &  \bft\oplus \bft \ar[l,"\id"'] \ar[r,"\nabla"] &  \bft 
\end{tikzcd}
\hspace{15 mm}
\begin{tikzcd}
\zero & \zero \ar[l,"\id"'] \ar[r,"!"] &  \bft 
\end{tikzcd}
\]
and the comultiplication and counit are their transposes.

Let $\tau$ be a set, which we interpret as a set of types. We now specialize to the slice category $\cat{C}\coloneqq\Cat{Set}/\tau$, for which an object is a function $x\to\tau$, regarded as a \emph{$\tau$-typed set}; when we speak of its elements, we mean elements of $x$. Note that $\cat{C}\coloneqq\Cat{Set}/\tau$ is extensive, so the above discussion applies. In the setting of $\Cat{Sp}\cat{C}$ we will refer to elements of the domain and codomain of a map as \emph{ports}, to elements of the apex as \emph{wires}, and to the legs as \emph{attachment maps}.

The sum of two spans $
\begin{tikzcd}
\bft & \alpha \ar[l, "f\inp"'] \ar[r, "f\outp"] & \bfv
\end{tikzcd}
$ and $\begin{tikzcd}
\bft & \beta \ar[l, "g\inp"'] \ar[r, "g\outp"] & \bfv
\end{tikzcd}
$
is given by the span
\[
\begin{tikzcd}[column sep=large]
\bft & \alpha\oplus\beta \ar[l, "f\inp\fork g\inp"'] \ar[r, "f\outp\fork g\outp"]  & \bfv
\end{tikzcd}
\] 
where, given $f:\bft\to\bfv$ and $g:\bfu\to\bfv$, $f\triangledown g:\bft\oplus\bfu\to\bfv$ is the composition $\iota_{\bft}f+\iota_{\bfu}g$. We hence interpret this enrichment in commutative monoids as a disjunction: in $\alpha\oplus\beta$, the $\bft$ ports attach to the $\bfv$ ports via \emph{either} $\alpha$-wires \emph{or} $\beta$-wires. Of course, the zero span is $\bft\from\zero\to\bfv$, also denoted $\zero$.

In linear algebra, block matrices of linear maps correspond to direct sum decompositions, while ordinary matrices of \emph{scalars} correspond to maximal decompositions into one-dimensional subspaces. Similarly, in $\Cat{Sp}(\Cat{Set}/\tau)$, we have maximal decompositions into singleton sets. What serves as a basis in this context is simply a choice of ordering on a set, which we represent via enumeration as a tuple. In such a decomposition, every entry of the corresponding matrix represents a subspan whose domain and codomain are both singletons. In this case, the attachment maps are trivial, and hence we can, without loss of information, represent such entries via the set of wires that connect the domain and codomain singletons. Furthermore, for the sake of notational convenience, if a span $\Phi$ is given by a diagonal matrix, i.e. consists of a direct sum of singleton spans with apex sets $S_1,\dots,S_n$, we will write $S_1\oplus\cdots\oplus S_n$ for the total span.

\begin{example}
The $(r,s,t)\to(x,y,z)$ matrix 
\[
\begin{bmatrix}
\zero & \{B\} & \zero \\ 
\zero & \zero & \zero \\
\{C\} & \zero & \{A,D\}
\end{bmatrix}
\]
represents the span
$
\begin{tikzcd}
\{r,s,t\} &  \{A,B,C,D\} \ar[l, "f\inp"'] \ar[r, "f\outp"] & \{x,y,z\}
\end{tikzcd}
$
given by
\begin{align*}
    f\inp(A)  &= t
&   f\inp(B)  &= r
&   f\inp(C)  &= t
&   f\inp(D)  &= t
\\
    f\outp(A) &= z 
&   f\outp(B) &= y
&   f\outp(C) &= x 
&   f\outp(D) &= z
\end{align*}
We note that the span matrix is like an adjacency matrix, except that rather than merely indicating the presence of connection via a truth value, the entries indicate the set of all such connections.
\end{example}
We note that the composite of the span $\{s\}\leftarrow\alpha\to\{t\}$ and the span $\{t\}\leftarrow\beta\to\{v\}$ is the span $\{s\}\leftarrow\alpha\times\beta\to\{v\}$. We can hence perform matrix multiplication as in the following example.

\begin{example}\label{ex.composite_span}
The composite of the span matrices
\[
\begin{tikzpicture}
\node (A) at (-6,0) {$(j,k)$};
\node (B) at (0,0) {$(r,s,t)$};
\node (C) at (6,0) {$(x,y,z)$};

\node[above] at (-3,0) {\small $
\begin{bmatrix}
\zero & \{N\} & \{L, M\} \\ 
\{O,P\} & \zero & \{Q\}
\end{bmatrix}
$};
\node[above] at (3,0) {\small $
\begin{bmatrix}
\{F,G\}  & \{B\} & \zero \\ 
\zero & \zero & \{A,D\} \\
\{C\} & \zero & \{E\}
\end{bmatrix}
$};

\draw [->] (A) -- (B);
\draw [->] (B) -- (C);
\end{tikzpicture}
\]
is given by the following $(j,k)\to(x,y,z)$ matrix
\[
\begin{bmatrix}
\{L,M\}\times\{C\} & \zero & \{N\}\times\{A,D\}+\{L,M\}\times\{E\}\\
\{O,P\}\times\{F,G\} + \{Q\}\times\{C\} & \{O,P\}\times\{B\} & \{Q\}\times\{E\}\\
\end{bmatrix}
\]
\end{example}

Given our interpretation of apexes as wires, allowing non-bijective attachment maps corresponds to allowing wires to split, merge, terminate, and initiate. For the purposes of defining wiring diagrams for (strict) symmetric monoidal categories, we will (in Definition~\ref{def:wd_operad}) restrict the legs of our spans to bijections; however, the more general definition can be readily used to define wiring diagrams for SMCs in which objects are supplied with monoids and/or comonoids. In particular, this bijectivity restriction is lifted in Listing~\ref{lst:substitute-wires}, where ports can attach to multiple other ports on either side.

We denote the category of $\tau$-typed finite sets and typed bijections between them as $\bij_\tau$ and the associated category of spans by $\Cat{Sp}(\bij_\tau)$. One might remark that spans of bijections are equivalent simply to bijections and wonder why Definition~\ref{def:wd_operad} uses the former. Indeed, since $\Cat{Sp}(\bij_\tau)$ is equivalent to $\bij_\tau$, we could have in principle simply defined such wiring diagrams in the latter category. However, our composition formula mirrors the case analysis of Listing~\ref{lst:substitute-wires} by leveraging the matrix calculus of the biproduct category $\Cat{Sp}(\Cat{Set}/\tau)$.

\section{The operad of acyclic wiring diagrams}\label{ch:operads}

We are now equipped to characterize unbiased compositions of morphisms in a strict symmetric monoidal category. We make the strictness assumption for a couple of reasons. First, any monoidal category is monoidally equivalent to a strict one and hence no loss of generality is incurred. Second, the graphical languages of string and wiring diagrams prohibits bracketing of parallel wires and hides unit wires, corresponding to strict associativity and unitality.

In addition to requiring that our symmetric monoidal category be strict, we define the monoidal product in an unbiased manner, as in \cite{deligne1982}. Thus, we index the hom-sets by a pair of typed finite sets---rather than typed finite ordinals---whose monoidal product give the domain and codomain respectively. In the context of computer science, this choice corresponds to a dictionary-like representation of domain and codomain rather than the more traditional list-like representation. We note that this does \emph{not} imply that we enforce commutativity of the monoidal product. For instance, the typed finite sets $f,g:\{a_1,a_2\}\rightrightarrows \tau$ given by $f(a_1)=t_1, f(a_2)=t_2$ and $g(a_1)=t_2, g(a_2)=t_1$ are distinct, though isomorphic.

Fix a set $\tau$, whose elements we think of as types. Next we will define the operad $\cat{W}_\tau$ of \emph{acyclic wiring diagrams}. Roughly speaking, morphisms in $\cat{W}_\tau$ specify compositions (serial, parallel, etc.) of morphisms in an arbitrary strict unbiased SMC $\cat{C}$ that has been equipped with a function $\tau\to\ob(\cat{C})$, i.e., for which an object of $\cat{C}$ has been chosen for each element of $\tau$.

\begin{definition}\label{def:wd_operad}
Let $\tau$ be a set, elements of which we call \emph{types}. We define the $\tau$-typed operad $\cat{W}_\tau$ of acyclic wiring diagrams as follows.
\begin{itemize}
    \item an object, called a \emph{box}, $\bft$ is a $\tau$-typed signed set; i.e.\ a pair $(\bft\inp,\bft\outp)$ of $\tau$-typed sets, where a $\tau$-typed set is an object of $\Cat{Set}/\tau$; we call $\bft\inp$ the \emph{inputs} and $\bft\outp$ the \emph{outputs}. 
    \item a morphism, called a \emph{wiring diagram}, $\Phi$ of type $\bft^1,\dots, \bft^n \to \bfv$ is a span in $\bij_{\tau}$
    \[
        \begin{tikzcd}
        \bfv\inp\oplus\bft\outp  & \omega \ar[l,"\Phi\src"'] \ar[r,"\Phi\tgt"] & \bft\inp \oplus\bfv\outp 
        \end{tikzcd}
    \]
    where $\bft_\pm\coloneqq\bigoplus_{i=1}^n\bft^i_\pm$. We call $\omega$-elements \emph{wires} and enforce the following:
    \begin{itemize}
        \item[$\star$] \emph{progress condition}: imposing $\bft^i\prec\bft^j$ whenever there is a wire $A\in \omega$ for which both $\Phi\src(A)\in\bft^i\outp$ and $\Phi\tgt(A)\in\bft^j\inp$, the result must be a partial order on the $\bft^i$.
    \end{itemize} 
    \item the identity morphism $\inert_{\bft}:\bft\to\bft$ is given by the identity span
    \[
    \begin{tikzcd}
    \bft\inp\oplus \bft\outp & \bft\inp \oplus \bft\outp \ar[l, "\one"'] \ar[r,"\one"] & \bft\inp \oplus \bft\outp
    \end{tikzcd}
    \]
    we note that the block matrix form in the summands is given by the identity matrix.
    \item given wiring diagrams
$\Psi:\bfs^{i_1},\dots,\bfs^{i{}_{m_i}}\to\bft^i$
and $\Phi:\bft^1,\dots,\bft^n \to\bfv 
$,
    we now define their $i^{\text{th}}$ partial composite (sometimes called ``circle-$i$'' composition) \[\Psi\then_i\Phi:\bft^1,\dots,\bfs^{i_1},\dots,\bfs^{i{}_{m_i}},\dots\bft^n\to\bfv\]
    Letting $\bft^{\neg i}_\pm\coloneqq\bigoplus{}_{j\neq i}\bft^j_\pm$, we make the following abbreviations
    \begin{equation*}
    \bfu\inp \coloneqq 
        \bfv\inp \oplus \bft^{\neg i}\outp
    \hspace{15 mm} \bfu\outp \coloneqq 
       \bft^{\neg i}\inp\oplus\bfv\outp
\end{equation*}
    The composite $\Psi\then_i\Phi$ is then given by the block matrix of type $\bfu\inp\oplus\bfs\outp\to\bfs\inp\oplus\bfu\outp$
    \begin{equation}\label{eqn.composite_formula}
    \Psi\then_i\Phi = 
    \begin{bmatrix} 
    ({}_{\bfu\inp}\Phi{}_{\bft^i\inp})
    ({}_{\bft^i\inp}\Psi{}_{\bfs\inp})
    &
    {}_{\bfu\inp}\Phi{}_{\bfu\outp} 
    + 
    ({}_{\bfu\inp}\Phi{}_{\bft^i\inp})
    ({}_{\bft^i\inp}\Psi{}_{\bft^i\outp})
    ({}_{\bft^i\outp}\Phi{}_{\bfu\outp})
    \\
    {}_{\bfs\outp}\Psi{}_{\bfs\inp}
    &
    ({}_{\bfs\outp}\Psi{}_{\bft^i\outp})
    ({}_{\bft^i\outp}\Phi{}_{\bfu\outp})
    \end{bmatrix} 
    \end{equation}
If $\prec^\Phi$ and $\prec^\Psi$ are the orderings induced on the $\bft^i$ and $\bfs^j$ by $\Phi$ and $\Psi$ respectively, then one can define an ordering $\prec$ on $\{\bft^1,\dots,\bfs^{k_1},\dots,\bfs^{k{}_{m_k}},\dots\bft^n\}$ as follows
\begin{align*}
    \bfs^{k_i} \prec \bfs^{k_j} &\coloneqq \bfs^{k_i} \prec^\Psi \bfs^{k_j} & 
    \bfs^{k_i} \prec \bft^j &\coloneqq \bft^k \prec^\Phi \bft^j \\
    \bft^j \prec \bfs^{k_i} &\coloneqq \bft^j \prec^\Phi \bft^k
    &
    \bft^i \prec \bft^j &\coloneqq \bft^i \prec^\Phi \bft^j
\end{align*}
Thus $\Psi\then_i\Phi$ satisfies the progress condition.
\end{itemize}
\end{definition}

We remark that the composite formula \eqref{eqn.composite_formula} is a declarative version of the imperative substitution algorithm presented in Listing~\ref{lst:substitute-wires}. In particular, the four cases in the algorithm correspond to the four entries in the composite matrix: incoming wire, passing wire, fully internal wire, and outgoing wire. Furthermore, the products within each component, e.g. $({}_{\bfu\inp}\Phi{}_{\bft^i\inp})
({}_{\bft^i\inp}\Psi{}_{\bft^i\outp})
({}_{\bft^i\outp}\Phi{}_{\bfu\outp})$, correspond to the fusing together of wires. Finally, the sum in the upper right-hand entry codifies the fact that the composite wiring diagram still includes all of the wires, ${}_{\bfu\inp}\Phi{}_{\bfu\outp}$, that did not interact with input box $\bft_i$.

Each of the entries in the fully decomposed matrix representation of a bijective span as in Definition~\ref{def:wd_operad} is either the empty set or a singleton; we will write $\zero$ in the former case and write $A$ for the singleton set $\{A\}$. Similarly, we will write $(A,B)$ for $\{A\}\times\{B\}$.

\begin{example}\label{ex.smc_ops}
We now define wiring diagrams corresponding to core SMC operations. We call these \emph{$\cat{W}$-representations} for any set $\tau$; though we introduce them in an example, they will play an important role in the theory.

\begin{itemize}

    \item \textbf{Symmetry.} Consider wires $\omega$, no inner boxes, and outer box $\bft=(\omega,\omega)$. For any permutation $\sigma:\omega\xrightarrow{\sim}\omega$, the $\cat{W}$-representation of symmetry is given by the 0-ary wiring diagram $\sym^\sigma_\omega:()\to\bft$ represented by the permutation matrix of $\sigma$. We let $\unit_\omega\coloneqq\sym_\omega^\id$. For instance here is the wiring diagram and matrix form in the case where $\sigma$ is the transposition of two elements.
    
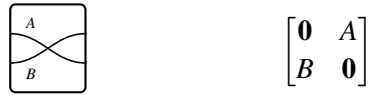
\[
    \begin{tikzpicture}[oriented WD, bb port length=0, font=\tiny, baseline=(B.north)]
        \node[bb={2}{2}] (outer) {};
        \draw (outer_in1) to node[above, pos=.25] {$A$} (outer_out2);
        \draw (outer_in2) to node[below, pos=.25] (B) {$B$} (outer_out1);    \end{tikzpicture}
        \hspace{1in}
    \begin{bmatrix}
    \zero & A
    \\
    B & \zero 
    \end{bmatrix}
    \]
    
    \item \textbf{Sequential composition.} Consider wires $\omega=\{A,B,C\}$, inner boxes $\bft,\bft'$, and outer box $\bfv$. The $\cat{W}$ representation of sequential composition (shown left) is given by the wiring diagram $\seq_{(A,B,C)}:\bft,\bft'\to\bfv$ with matrix of type $\bfv\inp\oplus\bft\outp\oplus\bft'\outp\to \bft\inp\oplus\bft'\inp\oplus\bfv\outp$ shown right:
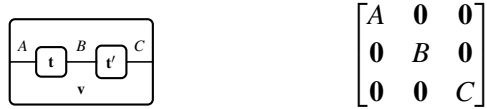
\[
\begin{tikzpicture}[oriented WD, bb small, bb port length=0, font=\tiny, baseline=(x1)]
      \node[bb={1}{1}] (x1) {$\bft$};
      \node[bb={1}{1}, right=of x1] (x2) {$\bft'$};
      \node[bb={1}{1}, inner xsep=10pt, inner ysep=10pt, fit=(x1) (x2)] (outer) {};
      \node[above] at (outer.south) {$\bfv$};
      \draw (outer_in1') -- node[above] {$A$} (x1_in1);
      \draw (x1_out1) -- node[above] {$B$} (x2_in1);
      \draw (x2_out1) -- node[above] {$C$} (outer_out1');
    \end{tikzpicture}
\hspace{1in}
    \begin{bmatrix}
    A & \zero & \zero  
    \\
    \zero & B & \zero 
    \\ 
    \zero & \zero & C
    \end{bmatrix}
\]

    \item \textbf{Parallel composition.} Consider wires $\omega=\{A,A',B,B'\}$, inner boxes $\bft,\bft'$, and outer box $\bfv$. The $\cat{W}$ representation of parallel composition is given by the wiring diagram $\para_{\tiny \left[ \begin{smallmatrix}A&B\\A'&B'\end{smallmatrix}\right]}:\bft,\bft'\to\bfv$ with matrix of type $\bfv\inp\oplus\bft\outp\oplus\bft'\outp\to \bft\inp\oplus\bft'\inp\oplus\bfv\outp$,  where both $\bfv\inp$ and $\bfv\outp$ are two-dimensional.
    
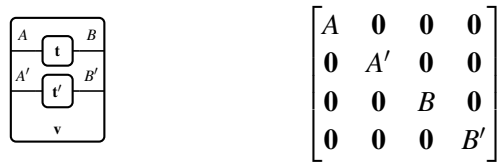
\[
\begin{tikzpicture}[oriented WD, bb small, font=\tiny, baseline=(t')]
    \node[bb={1}{1}] (t) {$\bft$};
    \node[bb={1}{1}, below=1 of t] (t') {$\bft'$};
    \node[bb={0}{0}, inner xsep=12pt, inner ysep=13pt, fit= (t.-5) (t')] (outer) {};
    \node[above =-1pt] at (outer.south) {$\bfv$};
    \draw (t_in1) -- node[above] {$A$} (outer.west|-t_in1);
    \draw (t_out1) -- node[above] {$B$}  (outer.east|-t_out1);
    \draw (t'_in1) -- node[above] {$A'$} (outer.west|-t'_in1);
    \draw (t'_out1) -- node[above] {$B'$} (outer.east|-t'_out1);
\end{tikzpicture}
\hspace{1in}
    \begin{bmatrix}
    A & \zero & \zero & \zero 
    \\
    \zero & A' & \zero & \zero 
    \\
    \zero & \zero & B & \zero  
    \\ 
    \zero & \zero & \zero & B'
    \end{bmatrix}
\]
\end{itemize}
\end{example}

\begin{example}
     \textbf{Non-example.} Consider the following wiring diagram $\Phi$, which has a loop. 
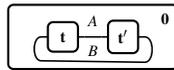
\[
\begin{tikzpicture}[oriented WD, bb small, font=\tiny]
    \node[bb={1}{1}] (x1) {$\bft$};
    \node[bb={1}{1}, right=of x1] (x2) {$\bft'$};
    \node[bb={0}{0}, inner xsep=15pt, inner ysep=6pt, fit=(x1) (x2)] (outer) {};
    \node[below left =-1pt] at (outer.north east) {$\zero$};
    \draw (x1_out1) -- node[above, font=\tiny] {$A$} (x2_in1);
    \draw let   
        \p2 = (x2.south east),
        \p1 = (x1.south west),
        \n1 = \bbportlen,
        \n2 = \bby
    in
        (x2_out1) to[in=0]
        (\x2+\n1,\y2-\n2) -- node[above=-2pt, font=\tiny] {$B$}
        (\x1-\n1,\y1-\n2) to [out=180] (x1_in1);
\end{tikzpicture}
\]
    This wiring diagram is of type $\bft,\bft'\to\zero$, with wires $\{A,B\}$, given by the span
    \begin{align*}
        \Phi\src(A) &= \bft\outp
    &   \Phi\tgt(A) &= \bft'\inp 
    &  \Phi\src(B) &= \bft'\outp 
    &   \Phi\tgt(B) &= \bft\inp 
    \end{align*}
    The ordering $\prec$ from Definition~\ref{def:wd_operad} is not a partial order: both $\bft\prec\bft'$ and $\bft'\prec\bft$. Therefore this diagram fails to satisfy the progress condition of Definition~\ref{def:wd_operad}.
\end{example}

\begin{example}\label{ex.smcs_wds} \textbf{Interchange Law.}
To demonstrate composition, we will prove that $\cat{W}$-representations satisfy the interchange law. 
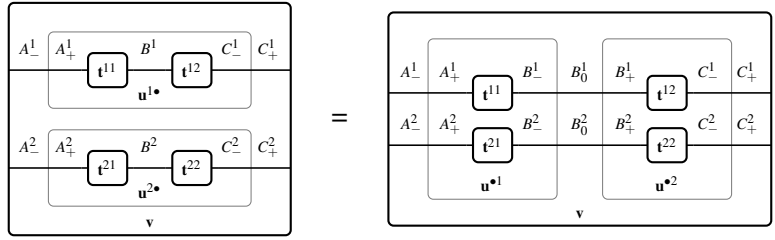
\[
\begin{tikzpicture}
\node (p1) {
\begin{tikzpicture}[oriented WD, bb small, font=\tiny]
    \node[bb={1}{1}] (t11) {$\bft^{11}$};
    \node[bb={1}{1}, right=1.5 of t11] (t12) {$\bft^{12}$};
    \node[bb={1}{1}, below=7 of t11] (t21) {$\bft^{21}$};
    \node[bb={1}{1}, right=1.5 of t21] (t22) {$\bft^{22}$};
    \node[bb={0}{0}, inner xsep=15pt, inner ysep=8pt, thin, gray, fit=(t11) (t12)] (t1) {};
    \node[bb={0}{0}, inner xsep=15pt, inner ysep=8pt, thin, gray, fit=(t21) (t22)] (t2) {};
    \node[bb={0}{0}, inner xsep=15pt, inner ysep=11pt, fit=(t1) (t2)] (outer) {};
    \node[above=-1pt] at (outer.south) {$\bfv$};
    \node[above=-1pt] at (t1.south) {$\bfu^{1\bullet}$};
    \node[above=-1pt] at (t2.south) {$\bfu^{2\bullet}$};
    \draw (t11_out1) -- node[above, font=\tiny] {$B^1\vphantom{A^1_+}$} (t12_in1);
    \draw (t21_out1) -- node[above, font=\tiny] {$B^2\vphantom{A^1_+}$} (t22_in1);
        \draw (t11_in1) -- node[above, font=\tiny, pos=.2] {$A^1\outp$} node[above, font=\tiny, pos=.7] {$A^1\inp$} (t11_in1-|outer.west);
        \draw (t21_in1) -- node[above, font=\tiny, pos=.2] {$A^2\outp$} node[above, font=\tiny, pos=.7] {$A^2\inp$} (t21_in1-|outer.west);
        \draw (t12_out1) -- node[above, font=\tiny, pos=.2] {$C^1\inp$} node[above, font=\tiny, pos=.7] {$C^1\outp$} (t12_out1-|outer.east);
        \draw (t22_out1) -- node[above, font=\tiny, pos=.2] {$C^2\inp$} node[above, font=\tiny, pos=.7] {$C^2\outp$} (t22_out1-|outer.east);
\end{tikzpicture}
};
\node (p2) [right=1 of p1] {
\begin{tikzpicture}[oriented WD, bb small, font=\tiny]
    \node[bb={1}{1}] (t11) {$\bft^{11}$};
    \node[bb={1}{1}, right=4.5 of t11] (t12) {$\bft^{12}$};
    \node[bb={1}{1}, below=2 of t11] (t21) {$\bft^{21}$};
    \node[bb={1}{1}, right=4.5 of t21] (t22) {$\bft^{22}$};
    \node[bb={0}{0}, inner xsep=17pt, inner ysep=14pt, thin, gray, fit=(t11) (t21)] (t1) {};
    \node[bb={0}{0}, inner xsep=17pt, inner ysep=14pt, thin, gray, fit=(t12) (t22)] (t2) {};
    \node[bb={0}{0}, inner xsep=15pt, inner ysep=10pt, fit=(t1) (t2)] (outer) {};
    \node[above=-1pt] at (outer.south) {$\bfv$};
    \node[above=-1pt] at (t1.south) {$\bfu^{\bullet1}$};
    \node[above=-1pt] at (t2.south) {$\bfu^{\bullet2}$};
    \draw (t11_out1) -- node[above, font=\tiny] {$B^1_0\vphantom{A^1_0}$} node[above, font=\tiny, pos=.12] {$B^1\inp\vphantom{A^1_0}$} node[above, font=\tiny, pos=.88] {$B^1\outp\vphantom{A^1_0}$} (t12_in1);
    \draw (t21_out1) -- node[above, font=\tiny] {$B^2_0\vphantom{A^1_0}$}  node[above, font=\tiny, pos=.12] {$B^2\inp\vphantom{A^1_0}$} node[above, font=\tiny, pos=.88] {$B^2\outp\vphantom{A^1_0}$} (t22_in1);
        \draw (t11_in1) -- node[above, font=\tiny, pos=.2] {$A^1\outp\vphantom{A^1_0}$} node[above, font=\tiny, pos=.7] {$A^1\inp\vphantom{A^1_0}$} (t11_in1-|outer.west);
        \draw (t21_in1) -- node[above, font=\tiny, pos=.2] {$A^2\outp\vphantom{A^1_0}$} node[above, font=\tiny, pos=.7] {$A^2\inp\vphantom{A^1_0}$} (t21_in1-|outer.west);
        \draw (t12_out1) -- node[above, font=\tiny, pos=.2] {$C^1\inp\vphantom{A^1_0}$} node[above, font=\tiny, pos=.7] {$C^1\outp\vphantom{A^1_0}$} (t12_out1-|outer.east);
        \draw (t22_out1) -- node[above, font=\tiny, pos=.2] {$C^2\inp\vphantom{A^1_0}$} node[above, font=\tiny, pos=.7] {$C^2\outp\vphantom{A^1_0}$} (t22_out1-|outer.east);
\end{tikzpicture}
};
\node[font=\normalsize] at ($(p1.east)!.5!(p2.west)$) {=};
\end{tikzpicture}
\]

We wish to show the following diagram commutes:
    \[
    \begin{tikzcd}[column sep = 55pt]
    \bft^{11},\bft^{21},\bft^{12},\bft^{22} 
        \ar[r,"\seq_{\omega^1}{,}\seq_{\omega^2}"]
        \ar[d,"\para_{\eta^1}{,}\para_{\eta^2}"']
    & 
    \bfu^{1\bullet},\bfu^{2\bullet}
        \ar[d,"\para_\eta"]
    \\
    \bfu^{\bullet1},\bfu^{\bullet2}
        \ar[r,"\seq_\omega"']
    &
    \bfv
    \end{tikzcd}
    \] 
    where $\bft^{11},\bft^{12},\bft^{21},\bft^{22}$ are inner boxes, $\bfu^{1\bullet},\bfu^{2\bullet},\bfu^{\bullet1},\bfu^{\bullet2}$ are intermediary boxes, $\bfv$ is the outer box, and the subscripts on the morphisms correspond to the following wire tuples:
    \begin{align*}
        \omega^1 &= (A^1\outp,B^1,C^1\inp)
    &   \omega^2 &= (A^2\outp,B^2,C^2\inp)   
    & \omega &= 
        (A^1\inp\oplus A^2\inp,B^1_0\oplus B^2_0,C^1\outp\oplus C^2\outp)
    \\
    \eta^1 &= 
        \begin{bmatrix}
            A^1\outp & B^1\inp
        \\  A^2\outp & B^2\inp  
        \end{bmatrix} 
    & \eta^2 &= 
        \begin{bmatrix}
            B^1\outp & C^1\inp
        \\  B^2\outp & C^2\inp  
        \end{bmatrix}
    &   \eta &= 
        \begin{bmatrix}
            A^1\inp & C^1\outp
        \\  A^2\inp & C^2\outp
        \end{bmatrix}
    \end{align*}
    The composite spans are represented as matrices of type
    \[\bfv\inp\oplus\bft^{11}\outp\oplus\bft^{21}\outp\oplus\bft^{12}\outp\oplus\bft^{22}\outp\to\bft^{11}\inp\oplus\bft^{21}\inp\oplus\bft^{12}\inp\oplus\bft^{22}\inp\oplus\bfv\outp\]
    where $\bfv_\pm$ are two-dimensional; we now compute these composites to be as follows:
    \begin{align*}
    \seq_{\omega^1},\seq_{\omega^2}\then\para_\eta &=\scriptsize
    \begin{bmatrix}
    (A^1\inp,A^1\outp) &  \zero & \zero & \zero & \zero & \zero 
    \\ 
   \zero &  (A^2\inp,A^2\outp)  & \zero & \zero & \zero & \zero 
    \\
    \zero & \zero &  B^1 & \zero & \zero & \zero 
    \\
    \zero & \zero & \zero & B^2 & \zero & \zero 
    \\
    \zero & \zero & \zero & \zero & (C^1\inp,C^1\outp) & \zero 
    \\
    \zero & \zero & \zero & \zero & \zero & (C^2\inp,C^2\outp)
    \end{bmatrix}
    \\
    \para_{\eta^1},\para_{\eta^2}\then\seq_\omega &=\scriptsize
    \begin{bmatrix}
    (A^1\inp,A^1\outp) &  \zero & \zero & \zero & \zero & \zero 
    \\ 
   \zero &  (A^2\inp,A^2\outp)  & \zero & \zero & \zero & \zero 
    \\
    \zero & \zero &  (B^1\inp,B^1_0,B^1\outp) & \zero & \zero & \zero 
    \\
    \zero & \zero & \zero & (B^2\inp,B^2_0,B^2\outp)  & \zero & \zero 
    \\
    \zero & \zero & \zero & \zero & (C^1\inp,C^1\outp) & \zero 
    \\
    \zero & \zero & \zero & \zero & \zero & (C^2\inp,C^2\outp)
    \end{bmatrix}
    \end{align*}
    These matrices of sets are isomorphic since they have isomorphic sets (singleton and empty sets) in corresponding entries.

\end{example}

\section{Operad algebras and symmetric monoidal categories}
\label{ch:smcs}

With the operad $\cat{W}_\tau$ of acyclic wiring diagrams (whose wires are labeled by a chosen set $\tau$) in hand, we now formalize how it captures the compositional structure of symmetric monoidal categories. Recall that a $\cat{W}_\tau$-algebra is an operad functor:
\[H\colon\cat{W}_\tau\to\smset,\]
where $\smset$ is the operad corresponding to the symmetric monoidal category $(\smset,\times,1)$. We will now show precisely how such algebras correspond to strict SMC's. Before doing so, we will first characterize $\cat{W}_\tau$ in terms of generators.

\begin{lemma}\label{lemma.SMC_W-alg}
The operad $\cat{W}_\tau$ is generated by $\sym,\seq,\para$ defined in Example~\ref{ex.smc_ops}, and these satisfy the usual axioms (parallel and sequential unitality and associativity, interchange, and permutation).
\end{lemma}
\begin{proof}[Sketch of proof]
We proceed by induction on the number of inner boxes. Let $\Phi\colon \bft_1,\ldots,\bft_n\to\bfv$ be a morphism. If $n=0$ then this is just some wire permutation $\sym_\sigma$. For $n\geq 0$, the progress condition gives a partial ordering on boxes. Without loss of generality, suppose that $\bft_1$ is minimal. Then we can rewrite the diagram in the form
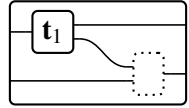
\begin{equation}\label{eqn.dmitry_special}
\begin{tikzpicture}[oriented WD, bb small]
 	\node[bb={1}{2}] (Monf) {$\bft_1$};
  	\node[bb={2}{1}, dotted, below right=0 and 2 of Monf] (Mong) {\vphantom{$\bft_1$}};
  	\node[bb={0}{0}, inner xsep=.3cm, fit=(Monf) (Mong)] (Mon) {};
  	\node[coordinate] at (Mon.west|-Monf_in1) (Mon_in1) {};
  	\node[coordinate] at (Mon.west|-Mong_in2) (Mon_in2) {};
  	\node[coordinate] at (Mon.east|-Monf_out1) (Mon_out1) {};
  	\node[coordinate] at (Mon.east|-Mong_out1) (Mon_out2) {};
  	\draw (Mon_in1) --  (Monf_in1);
  	\draw (Monf_out1) -- (Mon_out1);
  	\draw (Monf_out2) to (Mong_in1);
  	\draw (Mon_in2) --  (Mong_in2);
  	\draw (Mong_out1) -- (Mon_out2);
 \end{tikzpicture}
\end{equation}
where all the wires shown can represent multiple wires (possibly none). The wiring diagram in \eqref{eqn.dmitry_special} can be written as a sequential composite of parallel composites.
\erase{
Let the $f$ box be $\bft$, the dotted box be $\bfu$, and the outer box be $\bfv$. Inside the dotted box imagine a 2-box wiring diagram, where we name those two boxes $\bfs$ and $\bfs'$.

The outer wiring diagram attaches to the $\bfu\inp$ ports in two ways: from $\bft\outp$ ports and from $\bfv\inp$ ports, giving a map $outer:\bfu\inp\to\mathbf{2}$. The inner wiring diagram attaches to the $\bfu\inp$ ports in two ways: to $\bfs\outp$ ports and to $\bfs'\outp$ ports, giving another map $inner:\bfu\inp\to\mathbf{2}$. We can pair (fork) these maps to get a map $\bfu\inp\to\mathbf{2}\times\mathbf{2}$.
}
An example axiom relating series and parallel composition was proved in Example~\ref{ex.smcs_wds}. 
\end{proof}

We now show how $\cat{W}$-algebras give rise to SMC's.

\begin{theorem}
\label{thm.wdalgto_smc}
Let $\tau$ be a set. There is a fully faithful functor from the category of $\cat{W}_\tau$-algebras to that of strict SMC's whose objects are $\tau$-typed finite sets and whose morphisms are identity-on-objects symmetric monoidal functors.
\end{theorem}
\begin{proof}[Sketch of proof]
Let $H\colon\cat{W}_\tau\to\smset$ be a functor. We need to define an SMC $(\cat{C},\otimes, I)$ with objects $\tau$-typed finite sets; in particular, we need to define the hom-set for a given pair of objects $(\bft\inp,\bft\outp)\in\ob(\cat{C})$. But since objects in $\cat{W}_\tau$ are exactly such pairs, we may simply use $H$:
\begin{equation}\label{eqn.defn_C_from_H}
    \cat{C}(\bft\inp,\bft\outp)\coloneqq H(\bft\inp,\bft\outp).
\end{equation}
We will obtain the identities, composition, and monoidal structure of $\cat{C}$, by applying $H$ to various wiring diagrams.

For any object $A\in\ob(\cat{C})$, consider the box $\bft=(A,A)$ and the 0-ary wiring diagram $\unit_A\colon ()\to \bft$ given by $
\begin{tikzpicture}[oriented WD, bb small, font=\tiny]
    \node[bb = {1}{1}] (box) {\vphantom{t}};
    \draw (box_in1') -- node[left, pos=0] {$A$} node [right, pos=1] {$A$} (box_out1');
\end{tikzpicture}
$. 
We obtain a function $H(\unit_A)\colon 1\to H(A,A)$, and we define the identity on $A$ to be the image of the unique element.

Given objects $A,B,C\in\ob(\cat{C})$, we need a function $\Hom(A,B)\times\Hom(B,C)\to\Hom(A,C)$. We obtain it by applying $H$ to the wiring diagram $\seq_{(A,B,C)}:(A,B),(B,C)\to(A,C)$, giving 
\begin{align*}\then_{A,B,C} &\coloneqq H(\seq_{(A,B,C)})\colon H(A,B)\times H(B,C)\to H(A,C).
\intertext{Similarly, given objects $A,A',B,B'\in\ob(\cat{C})$, parallel composition is defined by} \otimes_{A,A',B,B'}
&\coloneqq H\Big(\para_{\tiny \left[ \begin{smallmatrix}A&B\\A'&B'\end{smallmatrix}\right]}\Big)\colon H(A,B)\times H(A',B')\to H(A\otimes A',B\otimes B').
\end{align*}
These two compositions are respectively depicted below.

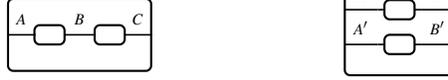
\[
\begin{tikzpicture}[oriented WD, bb small, bb port length=0, font=\tiny, baseline=(x1)]
  \node[bb={1}{1}] (x1) {};
  \node[bb={1}{1}, right=of x1] (x2) {};
  \node[bb={1}{1}, inner xsep=10pt, inner ysep=10pt, fit=(x1) (x2)] (outer) {};
  \draw (outer_in1') -- node[above] {$A$} (x1_in1);
  \draw (x1_out1) -- node[above] {$B$} (x2_in1);
  \draw (x2_out1) -- node[above] {$C$} (outer_out1');
\end{tikzpicture}
\hspace{1in}
\begin{tikzpicture}[oriented WD, bb small, font=\tiny, baseline=(t'.north)]
    \node[bb={1}{1}] (t) {};
    \node[bb={1}{1}, below=1.7 of t] (t') {};
    \node[bb={0}{0}, inner xsep=15pt, inner ysep=8pt, fit= (t) (t')] (outer) {};
    \draw (t_in1) -- node[above] {$A$} (outer.west|-t_in1);
    \draw (t_out1) -- node[above] {$B$}  (outer.east|-t_out1);
    \draw (t'_in1) -- node[above] {$A'$} (outer.west|-t'_in1);
    \draw (t'_out1) -- node[above] {$B'$} (outer.east|-t'_out1);
\end{tikzpicture}
\]

The interchange law was established in Example~\ref{ex.smcs_wds}, and the others, unitality and associativity of serial composition, and unitality, associativity, and symmetry of parallel composition, are similar.

It remains to show that for any two $\cat{W}_\tau$-algebras $H,H'$ with associated SMCs $\cat{C},\cat{C}'$, we have a bijection between the set of natural transformations $H\to H'$ and that of identity-on-objects symmetric monoidal functors $\cat{C}\to\cat{C}'$. A natural transformation $\alpha\colon H\to H'$ defines a function $H(\bft\inp,\bft\outp)\to H'(\bft\inp,\bft\outp)$, which by \eqref{eqn.defn_C_from_H} defines the required functor on hom-sets. This functor respects identity, composition, and the symmetric monoidal structure by the naturality of $\alpha$ with respect to the morphisms in $\cat{W}_\tau$ that correspond to these structures.
\end{proof}

The SMC resulting from the procedure in Theorem~\ref{thm.wdalgto_smc} is free-on-objects. In the next theorem we will prove that the functor from Theorem~\ref{thm.wdalgto_smc} is 2-essentially surjective: every symmetric monoidal category is equivalent to one coming from an $\cat{W}$-algebra.

\begin{theorem}\label{thm.smc_alg}
To every SMC $(\cat{C},\otimes,I)$ there is an associated $\cat{W}_{\ob(\cat{C})}$-algebra $H_{\cat{C}}$. If we then apply the construction of Theorem~\ref{thm.wdalgto_smc} to $H_\cat{C}$ to obtain an SMC $\cat{C}'$, there is an induced equivalence of symmetric monoidal categories $\cat{C}'\to\cat{C}$.
\end{theorem}
\begin{proof}[Sketch of proof]
Let $(\cat{C},\otimes,I)$ be a symmetric monoidal category. We define a functor  $H_{\cat{C}}\colon\cat{W}_{\ob(\cat{C})}\to\smset$ as follows. Suppose given a box, i.e.\ an object $(\bft\inp,\bft\outp)$ in $\cat{W}_{\ob(\cat{C})}$, where $\bft\inp\colon S\inp\to\ob(\cat{C})$ and $\bft\outp\colon S\outp\to\ob(\cat{C})$ are the typed finite sets. To it we assign the hom-set 
\begin{equation}\label{eqn.define_H}
H_{\cat{C}}(\bft\inp,\bft\outp)
\coloneqq
\cat{C}\left(
\bigotimes_{s\in S\inp}\bft\inp(s),
\bigotimes_{s\in S\outp}\bft\outp(s)
\right).
\end{equation}
By Lemma~\ref{lemma.SMC_W-alg}, $\cat{W}_{\ob(\cat{C})}$ is generated by the morphisms $\para,\seq$, and $\sym$ corresponding to parallel composition, series composition, and permutation, and that these morphisms satisfy well-known relations. Thus to give the action of $H_{\cat{C}}$ on morphisms, it suffices to say how it acts on these generators, and show that relations hold. For the morphism $\seq_{(A,B,C)}$ in $\cat{W}_{\ob(\cat{C})}$ representing series composition
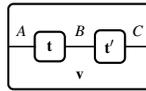
\[
\begin{tikzpicture}[oriented WD, bb small, bb port length=0, font=\tiny]
  \node[bb={1}{1}] (x1) {$\bft$};
  \node[bb={1}{1}, right=of x1] (x2) {$\bft'$};
  \node[bb={1}{1}, inner xsep=10pt, inner ysep=10pt, fit=(x1) (x2)] (outer) {};
  \node[above] at (outer.south) {$\bfv$};
  \draw (outer_in1') -- node[above] {$A$} (x1_in1);
  \draw (x1_out1) -- node[above] {$B$} (x2_in1);
  \draw (x2_out1) -- node[above] {$C$} (outer_out1');
\end{tikzpicture}
\]
we need to give a function $H_{\cat{C}}(A,B)\times H_{\cat{C}}(B,C)\to H_{\cat{C}}(A,C)$. By definition this is just a function $\cat{C}(A,B)\times\cat{C}(B,C)\to\cat{C}(A,C)$, and of course we use the composition function $\then_{A,B,C}$ from $\cat{C}$ as a category. The case for parallel composition (resp.\ symmetry) is similar: one uses the monoidal product (resp.\ symmetry) from $\cat{C}$. These satisfy the required diagrammatic relations because by definition $\cat{C}$ satisfies the laws of monoidal categories.

Finally we consider the roundtrip, starting with $\cat{C}$, constructing $H_\cat{C}$ as above, and then applying the construction of Theorem~\ref{thm.wdalgto_smc} to obtain a new SMC $\cat{C}'$. The objects of $\cat{C}'$ are the $\ob(\cat{C})$-typed finite sets, and there is a surjection $\ob(\cat{C}')\to\ob(\cat{C})$ given by sending $\bft\colon S\to\ob(\cat{C})$ to $\bigotimes_{s\in S}\bft(s)$. At this point, \eqref{eqn.defn_C_from_H} and \eqref{eqn.define_H} provide a bijection $\cat{C}'(\bft\inp,\bft\outp)\to\cat{C}(\bigotimes_{s\in S\inp}\bft\inp(s),\bigotimes_{s\in S\outp}\bft\outp(s))$, which one can check is part of a surjective-on-objects fully faithful symmetric monoidal functor.
\end{proof}

\subsection*{Acknowledgments}
David Spivak acknowledges support from AFOSR grants FA9550-19-1-0113 and FA9550-17-1-0058. The authors also appreciate the helpful reviews of the ACT2020 program committee, which improved the quality of this paper.

\printbibliography

\end{document}